\documentclass[a4paper]{article}
\usepackage[utf8]{inputenc}

\usepackage[margin=1.3in]{geometry}

\usepackage{amsmath,amssymb,amscd,amsthm,amsfonts}
\usepackage{graphicx,subcaption}
\usepackage{hyperref}
\usepackage{dsfont}
\usepackage{xcolor}
\usepackage{color,soul}
\usepackage[utf8]{inputenc}
\usepackage[capitalise]{cleveref}
\usepackage{thm-restate}

\newtheorem{theorem}{Theorem}[section]

\newtheorem{lemma}[theorem]{Lemma}

\newtheorem{conjecture}[theorem]{Conjecture}

\theoremstyle{definition}

\graphicspath{{figures/}}

\author{Johanna Ockenfels \\
        Dept. of Computer Science, ETH Z\"{u}rich, Switzerland \\ {\tt jockenfels@student.ethz.ch}
        \and
        Yoshio Okamoto \\
        Dept. of Computer and Network Engineering, \\The University of Electro-Communications, Japan \\ {\tt okamotoy@uec.ac.jp}
        \and
        Patrick Schnider \\
        Dept. of Mathematics and Computer Science,
        University of Basel\\
        Dept. of Computer Science, ETH Z\"{u}rich, Switzerland \\ {\tt patrick.schnider@inf.ethz.ch}
}

\title{Chasing puppies on orthogonal straight-line plane graphs\footnote{This work was initiated during the 21st Gremo's Workshop on Open Problems in Pura, Switzerland. The authors would like to thank the other participants for the lively discussions. The second author acknowledges the support by JSPS KAKENHI Grant Number JP23K10982.}}
\date{}

\begin{document}

\maketitle

\begin{abstract}
Assume that you have lost your puppy on an embedded graph. You can walk around on the graph and the puppy will run towards you at infinite speed, always locally minimizing the distance to your current position. Is it always possible for you to reunite with the puppy? We show that if the embedded graph is an orthogonal straight-line embedding the answer is yes.
\end{abstract}

\section{Introduction}

The puppy chasing problem was proposed by Michael Biro at the 2013 edition of the Canadian Conference on Computational Geometry. While inspired by problems in beacon-base routing and originally illustrated by train tracks and locomotives, like many other problems in discrete and computational geometry it has a dog-related illustration, which was proposed by the authors who solved Biro's original problem \cite{abrahamsen2022chasing}.

Assume you are walking with your puppy on an embedded graph. Suddenly, you realize that the puppy is not with you anymore. Luckily, you can see each other and the puppy also wants to reunite with you. While the puppy can run infinitely fast, its behavior is very naive: it runs as close to you as possible, always locally improving the distance to you. Is there always a route you can walk, so that the puppy will eventually reunite with you?

To phrase it in mathematical terms, let $G$ be a finite graph and let $\gamma\colon G\rightarrow\mathbb{R}^2$ be a (crossing-free) embedding of the graph. We interpret $G$ as a topological space. Any pair $(x,y)\in G\times G$ then defines locations $h=\gamma(x)$ and $p=\gamma(y)$ in the plane $\mathbb{R}^2$, which we interpret as the positions of the human and the puppy, respectively. Denoting by $d(x,y)$ the Euclidean distance between $h=\gamma(x)$ and $p=\gamma(y)$, we say that a configuration $(x,y)\in G\times G$ is \emph{stable} if there is some $\varepsilon>0$ such that for all $y'\in G$ with $d(y,y')\leq\varepsilon$ we have $d(x,y')\geq d(x,y)$. In other words, a configuration is stable if the puppy cannot locally decrease its distance to the human. For a non-stable configuration there is thus at least one direction in which the puppy can decrease its distance to the human, in which case it will do so. If there are several such directions, the puppy chooses an arbitrary one. See Figure \ref{fig:puppies_example}. 

\begin{figure}
    \centering
    \includegraphics[scale = 1]{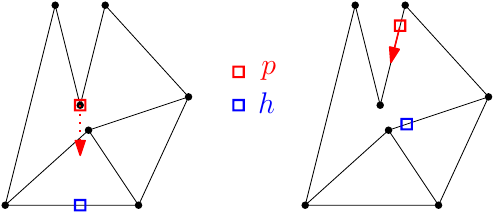}
    \caption{A stable configuration (left), and an unstable one (right).}
    \label{fig:puppies_example}
\end{figure}

If $G$ is a cycle, and the embedding is either smooth or piece-wise linear, it was shown by Abrahamsen, Erickson, Kostitsyna, L\"{o}ffler, Miltzow, Urhausen, Vermeulen and Viglietta (AEKLMUVV) that the human can always catch the puppy \cite{abrahamsen2022chasing}, solving Biro's original question. While there are embeddings where the human does not even need to move to catch the puppy, e.g., the unit circle, there are also examples where if moving the wrong way the human will never catch the puppy, see, e.g., Figure \ref{fig:star}. If self-intersections are allowed (think of one edge crossing over the other via a bridge), then there are drawings where the human can never catch the puppy. An example of this is the double loop depicted in Figure \ref{fig:double_circle}.

\begin{figure}
    \centering
    \includegraphics[scale = 1]{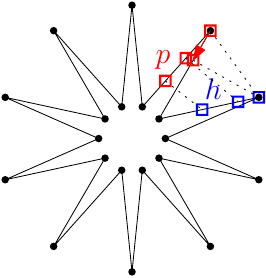}
    \caption{If the human $h$ continues walking clockwise, human and puppy $p$ will never be reunited.}
    \label{fig:star}
\end{figure}

\begin{figure}
    \centering
    \includegraphics[scale = 1]{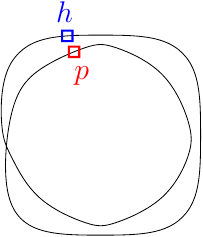}
    \caption{No matter how the human $h$ moves, human and puppy $p$ will never be reunited.}
    \label{fig:double_circle}
\end{figure}

It is an interesting open problem to characterize the drawings of a cycle on which the human can always catch the puppy. Towards this, Brunck, L\"{o}ffler and Silveira have recently shown that the rotation number does not give such a characterization by constructing a curve with rotation number 1 (i.e., the same rotation number as a crossing-free embedding), where the human sometimes cannot catch the puppy \cite{loffler4}.

\subsection{Our results}
In this work, we extend Biro's question from embeddings of circles to embeddings of general graphs. Concretely, motivated by the many results on beacon routing and related problems in orthogonal domains, we show the following theorem.

\begin{restatable}{theorem}{orthogonal_puppies}\label{thm:orthogonal_puppies}
Let $G$ be a finite connected graph and let $\gamma$ be an orthogonal straight-line embedding of $G$ in the plane. Then, there is a strategy for the human to catch the puppy on $\gamma$.
\end{restatable}

Here, an orthogonal straight-line embedding is an embedding of the graph where every edge is mapped to either a horizontal or a vertical line segment, see Figure \ref{fig:orthogonal_embedding} for an example.
Note that not every graph admits an orthogonal straight-line embedding, so we are implicitly restricting our attention to those that do. Also note that \emph{orthogonal drawings}, where edges can be drawn as polygonal lines with vertical and horizontal parts can be seen as orthogonal straight-line embeddings of a subdivision of $G$, so Theorem \ref{thm:orthogonal_puppies} immediately extends to orthogonal drawings.

\begin{figure}
    \centering
    \includegraphics[scale = 1]{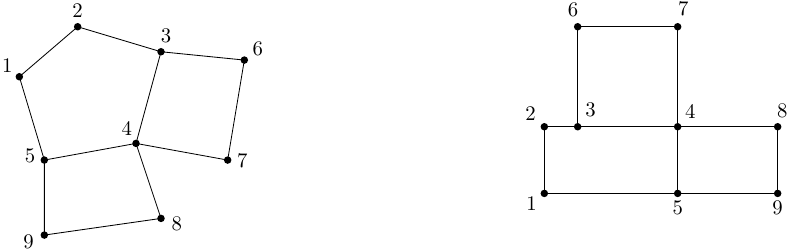}
    \caption{A straight-line embedding (left) and an orthogonal straight-line embedding (right) of the same graph.}
    \label{fig:orthogonal_embedding}
\end{figure}

Unlike the proof of AEKLMUVV that uses topological methods to show the existence of a strategy, we give a recursive strategy that is simple to describe. For orthogonal polygons (that is, orthogonal straight-line embeddings of a circle), AEKLMUVV also give a simple strategy (see \cite{abrahamsen2022chasing}, Theorem 2), which differs from our strategy. On the other hand, while we conjecture that Theorem \ref{thm:orthogonal_puppies} should hold for any straight-line embedding of a graph, our strategy heavily relies on orthogonality. We leave the following as an open problem.

\begin{conjecture}
Let $G$ be a finite graph and let $\gamma$ be a straight-line embedding of $G$ in the plane. Then, there is a strategy for the human to catch the puppy on $\gamma$.
\end{conjecture}

\subsection{Related work}
A natural variant of Biro's problem, which AEKLMUVV call the \emph{guppy} problem, was introduced by Kouhestani and Rappaport \cite{guppy}: the guppy is swimming in a simply connected lake, behaving just like the puppy in our problem at hand. The human, however, can only walk along the shore. Kouhestani and Rappaport conjectured that there is always a strategy for the human to catch the guppy. However, this conjectured was settled in the negative when Abel, Akitaya, Demaine, Demaine, Hesterberg, Korman, Ku and Lynch provided a counterexample consisting of an orthogonal polygon where no such strategy exists \cite{guppy_counterexample}.

As mentioned above, the puppy problem is motivated by the problem of beacon routing. A \emph{beacon} is a point which can be activated to create a magnetic pull. Given a polygonal domain~$P$, the question is now how many beacons need to be placed and activated sequentially to route a point in $P$ to a given goal location. The difference from the puppy problem is thus that beacons cannot move but that several of them can be used. Beacon routing has been studied extensively in the last 10 years, in particular for orthogonal polygonal domains \cite{beacon11,beacon10,beacon8,beacon3,beacon1,beacon2,beacon9,beacon5,beacon12,beacon7,beacon6,beacon4}.

\section{A Strategy for Orthogonal Embeddings}

In this section, we prove Theorem \ref{thm:orthogonal_puppies} in two steps. First, in Section \ref{sec:generic} we show it for the \emph{generic} case, where we assume that no two horizontal line segments of the embedding are on the same height. This restricted setting already highlights the main ideas but is simpler to describe. In Section \ref{sec:general} we then adapt our proof to the non-generic case.

\subsection{The generic case}\label{sec:generic}

We begin by proving the following lemma.
\begin{lemma}
\label{lemLower}
    Given a configuration of the human's position $h$ and the puppy's position $p$ and a horizontal line $l$ through $h$ where $p$ is not above $l$, the puppy will move above $l$ only if the human moves above $l$.
\end{lemma}
\begin{proof}
    Since all edges are either parallel or orthogonal to $l$, the only way for the puppy to move above $l$ is to take a vertical edge that crosses $l$. However, since the puppy is currently below or at the same level as the human, no movement of the human can lead to this: If the human moves horizontally, then taking a vertical edge that crosses $l$ can only increase the distance between the two. Similarly, if the human moves vertically without crossing $l$, the puppy will only increase its distance by crossing $l$.
\end{proof}

\begin{proof}[Proof of Theorem \ref{thm:orthogonal_puppies}]
    Let $\gamma$ be an orthogonal straight-line embedding of a graph $G$. We describe an algorithm for the human to catch the puppy. The basic idea is that we decrease the size of the part of the embedding containing the puppy in every step by considering parts of the domain as ``forbidden.'' In the following we assume without loss of generality that at the beginning the puppy is not above the human. The algorithm is the following, see Figure \ref{fig:strategy}.
    
    \begin{itemize}
        \item[1.] Move in the non-forbidden part to the (unique) topmost horizontal edge $e$. All (vertical) parts of edges that are above this position are now forbidden.
        \item[2.] If $e$ is a cut edge (i.e., bridge) of the non-forbidden part, move to the connected component of $\gamma \setminus\{e\}$ that contains the location of the puppy and consider the edge (as well as all other components of the graph) forbidden.
        \item[3.] Otherwise, simply leave the edge and consider it and everything above it forbidden from now on.
        \item[4.] Repeat this process in the non-forbidden part of the embedding.
    \end{itemize}

    \begin{figure}
    \centering
    \includegraphics[scale = 1]{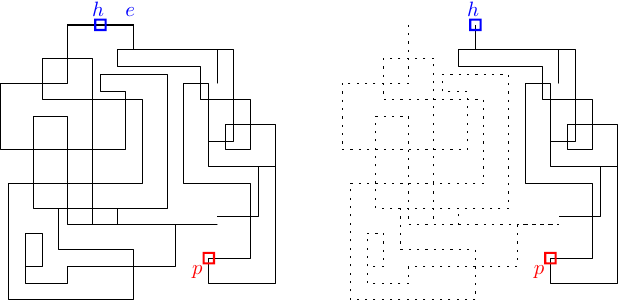}
    \caption{An example of Step 2 of the algorithm.}
    \label{fig:strategy}
\end{figure}

    We want to prove that whenever we label a new part of the domain as ``forbidden,'' neither the human nor the puppy will ever enter that region again. 
    We make an inductive argument. Our induction hypothesis is that before each iteration, the non-forbidden domain is connected and contains both the human and the puppy. Additionally, the forbidden domain has the property that the puppy will not enter any forbidden part if the human stays in the non-forbidden parts. These invariants are trivially true at the beginning of the process. 
    Now consider the configuration at the beginning of any iteration of the algorithm. By the induction hypothesis, the non-forbidden domain is connected and thus the human can move to the topmost edge $e$ without entering any forbidden parts. Again, by the induction hypothesis, the puppy is still in the non-forbidden domain when the human reaches $e$. 
    We argue that the induction hypothesis still holds after cutting off everything above $e$: by Lemma~\ref{lemLower}, the puppy cannot be above the human as $e$ is the topmost edge. Since we only cut off (parts of) vertical edges, the domain stays connected. Again by Lemma~\ref{lemLower}, the puppy will not enter the newly forbidden domain unless the human does.

  Now, there are two cases.
  \begin{enumerate}
    \item[1.] If $e$ is not a cut edge, then the human simply leaves the edge and considers it forbidden from now on. Now, the non-forbidden domain still contains the human and the puppy and is connected since otherwise $e$ must have been a cut edge. Moreover, since $e$ was the unique topmost horizontal edge and the human will never enter the forbidden parts again, by Lemma~\ref{lemLower}, the puppy cannot enter $e$ either.
    
    \item[2.] If $e$ is a cut edge, then the human moves to the connected component that contains the puppy and we forbid $e$ as well as all other connected components. After this, both the human and puppy are still in the non-forbidden domain. By definition, the domain is also still connected. By the same arguments as in the first case, the puppy will stay in the non-forbidden part.
  \end{enumerate}

    In each step the non-forbidden part of the domain decreases by at least one edge. Since there are only finitely many edges, the human and the puppy will eventually be located on the same edge at which point the human has caught the puppy.
\end{proof}

\subsection{Allowing horizontal edges of the same height}\label{sec:general}
Now, we are no longer assuming each horizontal edge to be at a different height, meaning that we need to adjust the strategy from the previous section as the topmost horizontal edge is no longer unique.
Let $m$ be the height of the topmost horizontal edge in $\gamma$, and let $T = \{e_1,\ldots,e_k\}$ be all edges in $\gamma$ at height $m$, ordered from left to right. Now, we consider the connected components of $\gamma\setminus T$ (note that this could be a single component). 
Our goal is to get to a configuration where the human is in the same component as the puppy. Once we are in such a configuration, we can again forbid all edges of $T$ and by Lemma~\ref{lemLower}, the puppy will never leave the component again, unless the human does, meaning that we can recursively continue the strategy on the non-forbidden part.

It thus remains to show that the human can always move into the same component as the puppy.
For each component $C$ of $\gamma\setminus T$ denote by $T(C)$ the set of edges in $T$ it is incident to.
We say that a component $C$ of $\gamma\setminus T$ is a \emph{U-component} if there is another component $C'$ that lies entirely inside the region bounded by $C$ and the horizontal line at height $m$. In this case we say that $C$ \emph{dominates} $C'$.
This gives a nesting of the components, see Figure \ref{fig:components} for an illustration.

\begin{figure}
    \centering
    \includegraphics[scale = 1]{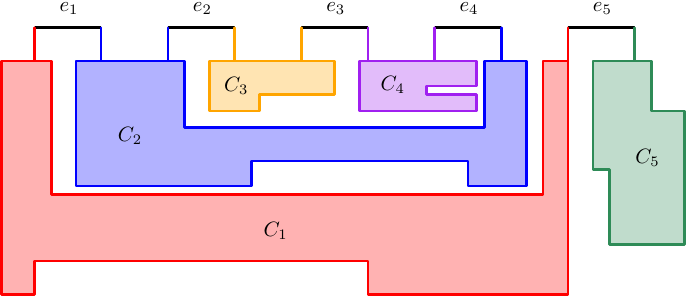}
    \caption{An orthogonal embedding with five components. The components $C_1$ and $C_2$ are U-components. The component $C_1$ dominates $C_2$, which in turn dominates $C_3$ and $C_4$.}
    \label{fig:components}
\end{figure}

\begin{lemma}\label{lem:u_components}
    Let $C_u$ be a U-component and let $\mathcal{D}$ be the set of all components dominated by $C_u$ (including $C_u$). If both human and puppy are on components in $\mathcal{D}$ and the human stays on components in $\mathcal{D}$, then the puppy also stays on components in $\mathcal{D}$.
\end{lemma}

\begin{proof}
    Let $e^\mathrm{l}$ and $e^\mathrm{r}$ be the leftmost and rightmost edges in $T(C_u)$, respectively. In order to leave $\mathcal{D}$, the puppy must either run leftwards over $e^\mathrm{l}$ or rightwards over $e^\mathrm{r}$. However, the puppy would only do this if the human was on an edge in $T$ left of $e^\mathrm{l}$ or right of $e^\mathrm{r}$, in which case the human is not in $\mathcal{D}$.
\end{proof}

Consider now the graph $G_\mathcal{D}$ whose vertices are the components that are not dominated by some other component, where two such components are connected whenever they are connected via an edge in $T$. We claim that $G_\mathcal{D}$ is a path: for each vertex $C_i$ of $G_\mathcal{D}$ let $e_i^\mathrm{l}$ and $e_i^\mathrm{r}$ denote the leftmost and rightmost edge in $T(C_i)$, respectively. The intervals 
enclosed by $e_i^\mathrm{l}$ and $e_i^\mathrm{r}$
must be pairwise disjoint by the above observations. Thus, there is a linear order along which the components are connected, implying that $G_\mathcal{D}$ is indeed a path.

The strategy of the human is now the following: walk to the vertex in $G_\mathcal{D}$ corresponding to the component $C$ dominating the component where the puppy currently is. As $G_\mathcal{D}$ is a path, this can always be achieved. Then, remove all components not dominated by $C$ from consideration. By Lemma \ref{lem:u_components}, the puppy will always stay on the remaining graph. If the puppy is also in $C$, remove everything but $C$ from consideration and recurse. Otherwise, walk into the connected component of $\gamma\setminus C$ in which the puppy is, remove everything else from consideration and recurse. By Lemma \ref{lem:u_components} and the above observations, the puppy will again stay on the remaining graph. As the size of the part of the graph embedding $\gamma$ under consideration decreases in each recursion step, at some point the human is in the same component $C$ of $\gamma\setminus T$. The human enters this component via an edge in $T$. Thus, if the puppy is also on an edge in $T$ it must be on the same one and it will run to the human. Otherwise the puppy is below the human at which point we can recurse by considering the topmost edges of $C$.

\bibliography{refs}

\end{document}